\documentclass[a4paper,UKenglish,cleveref, autoref, thm-restate]{oasics-v2021}

\usepackage[ruled]{algorithm}
\usepackage[noend]{algpseudocode}
\usepackage{amsthm, multicol}

\usepackage{amsmath}
\usepackage{graphicx}
\usepackage{xspace}
\usepackage{wrapfig, caption}

\floatname{algorithm}{{\small Code}}
\usepackage{enumitem}
\newcommand{\pr}{p}

\algblockdefx[Receive]{Receive}{EndReceive}%
[1]{{\bf receive}   (#1) from process $\pr$}%
{{\bf end receive}}

\algblockdefx[Receivex]{Receivex}{EndReceivex}%
[2]{{\bf receive}   (#1) from process #2}%
{{\bf end receive}}

\algblockdefx[Upon]{Upon}{EndUpon}%
[1]{{\bf upon} (#1) do}%
{{\bf end upon}}

\makeatletter
\ifthenelse{\equal{\ALG@noend}{t}}%
{\algtext*{EndReceive}}
{}%

\makeatletter
\ifthenelse{\equal{\ALG@noend}{t}}%
{\algtext*{EndUpon}}
{}%

\algrenewcommand{\ALG@beginalgorithmic}{\small}
\algrenewcommand\alglinenumber[1]{\small #1:}

\newcommand{\ar}[1]{\textcolor{cyan}{#1}}
\newcommand{\af}[1]{\textcolor{olive}{#1}}
\newcommand{\nn}[1]{\textcolor{blue}{#1}}

\newcommand{\cg}[1]{\textcolor{purple}{#1}}

\newcommand{\prf}[1]{{}}

\newtheorem{Def}{Definition}[section]

\newcommand{\sacode}[5]
{ \vspace{.06in} \hrule \vspace{.06in} 
 \noindent {\bf #1}: \\
 \footnotesize \noindent {\bf Signature:}\B \nobreak
 \normalsize \begin{quote} \nobreak #2 \end{quote}
 \footnotesize \noindent {\bf States:}\B \nobreak
 \begin{quote} \nobreak #3 \end{quote}
 \noindent {\bf Transitions:} \nobreak
 \vspace{-.2in} \nobreak
 \normalsize #4
 \vspace{-.06in} \hrule \vspace{.06in} 
}

\newcommand{\act}[1]{%
    \relax\ifmmode
        \mathord{\mathcode`\-="702D\sf #1\mathcode`\-="2200}%
    \else
        $\mathord{\mathcode`\-="702D\sf #1\mathcode`\-="2200}$%
    \fi
}

\newcommand{\tup}[1]{%
    \relax\ifmmode
      \langle #1 \rangle%
    \else
        $\langle$#1$\rangle$%
    \fi
}

\newcommand{\seq}[1]{%
    \relax\ifmmode
      \langle \! \langle #1 \rangle \! \rangle%
    \else
        $\langle \! \langle$ #1 $\rangle \! \rangle$%
    \fi
}

\newcommand{\B}{\vspace*{-\smallskipamount}}

\newcommand{\FF}{\vspace*{\medskipamount}}

\newcommand{\ms}[1]{%
    \relax\ifmmode
        \mathord{\mathcode`\-="702D\it #1\mathcode`\-="2200}%
    \else
        {\it #1}%
    \fi
}

\newcommand{\lit}[1]{%
    \relax\ifmmode
        \mathord{\mathcode`\-="702D\sf #1\mathcode`\-="2200}%
    \else
        {\it #1}%
    \fi
}

\newcommand{\XDK}[1]{}
\newcommand{\remove}[1]{} 
\newcommand{\proofremove}[1]{} 
\newcommand{\uselater}[1]{} 

\makeatletter
\def\mainlistofsymbols{
  \normalsize
  \vspace*{1.5 em}
  \@starttoc{los}
}

\def\partonelistofsymbols{
  \normalsize
  \vspace*{1.5 em}
  \@starttoc{p1los}
}

\def\parttwolistofsymbols{
  \normalsize
  \vspace*{1.5 em}
  \@starttoc{p2los}
}

\def\l@symbol#1#2{\addpenalty{-\@highpenalty} \vskip 4pt plus 2pt
{\@dottedtocline{0}{0em}{8em}{#1}{#2}}}
\makeatother

\newcommand{\newhiddensym}[2]{%
}

\newcommand{\algIOA}[2]{\ifmmode{\text{#1}_{#2}}\else{$\text{#1}_{#2}$}\fi}

\newcommand{\EX}{\ifmmode{\xi}\else{$\xi$}\fi}
\newcommand{\EXF}{\ifmmode{\phi}\else{$\phi$}\fi}

\newcommand{\inter}[1]{
	\ifmmode{\left(\bigcap_{\mathcal{Q}\in#1}\mathcal{Q}\right)}
	\else{$\left(\bigcap_{\mathcal{Q}\in#1}\mathcal{Q}\right)$}
	\fi
}

\newcommand{\idSet}{\mathcal{I}}

\mathchardef\mhyphen="2D

\renewcommand{\pr}{p}

\newcommand{\vid}[1]{\ifmmode{\nu_{#1}}\else{$\nu_{#1}$}\fi}

\newcommand{\seen}{\ifmmode{seen}\else{$seen$}\fi}

\newcommand{\maxts}[1]{\ifmmode{maxTS_{#1}}\else{$maxTS_{#1}$}\fi}
\newcommand{\maxtag}[1]{\ifmmode{maxTag_{#1}}\else{$maxTag_{#1}$}\fi}
\newcommand{\maxpair}[1]{\ifmmode{maxMPair_{#1}}\else{$maxMPair_{#1}$}\fi}
\newcommand{\mintag}[1]{\ifmmode{minTag_{#1}}\else{$minTag_{#1}$}\fi}
\newcommand{\maxps}{\ifmmode{maxPS}\else{$maxPS$}\fi}
\newcommand{\conftg}[1]{\ifmmode{confirmed_{#1}}\else{$confirmed_{#1}$}\fi}
\newcommand{\maxconftag}{\ifmmode{\ms{maxCT}}\else{$maxCT$}\fi}
\newcommand{\rread}{\act{read}\xspace}
\newcommand{\wwrite}{\act{write}\xspace}

\newcommand{\NACK}{\mathit{NACK}\xspace}
\newcommand{\ACK}{\mathit{ACK}\xspace}
\newcommand{\linearizable}{linearizable\xspace}
\newcommand{\regular}{regular\xspace}
\newcommand{\local}{local\xspace}
\newcommand{\order}{totally-ordered\xspace}
\newcommand{\valid}{\act{valid}\xspace}
\newcommand{\execute}{\act{execute}\xspace}
\newcommand{\apply}{\act{apply}\xspace}
\newcommand{\timestamp}{\mathit{ts}\xspace}

\newcommand{\pprec}{\prec\!\!\!\prec\xspace}

\newcommand{\ApplyIfValid}{\act{LoggedApply}}

\newcommand{\ibalance}{\mathit{ibalance}\xspace}
\newcommand{\transfer}{\act{transfer}\xspace}
\newcommand{\balancex}{\act{read}\xspace}

\title{Validated Objects: Specification, Implementation, and Applications}

\author{Antonio {Fernández Anta}}{IMDEA Networks Institute, Spain}{antonio.fernandez@imdea.org}{https://orcid.org/0000-0001-6501-2377}{Partially supported by Spanish State Research Agency (AEI) ECID project PID2019-109805RB-I00, cofunded by FEDER.}
\author{Chryssis Georgiou}{University of Cyprus, Cyprus}{chryssis@ucy.ac.cy}{https://orcid.org/0000-0003-4360-0260}{}
\author{Nicolas Nicolaou}{Algolysis Ltd, Cyprus}{nicolas@algolysis.com}{https://orcid.org/0000-0001-7540-784X}{}
\author{Antonio Russo}{IMDEA Networks Institute, Spain}{antonio@antoniorusso.me}{https://orcid.org/0000-0003-3795-3000}{}

\authorrunning{Fernández Anta et al.}



\hideOASIcs
\nolinenumbers

\begin{document}

\maketitle

\begin{abstract}
Guaranteeing the validity of concurrent operations
on distributed objects
is a key property for ensuring 
reliability and consistency 
in distributed systems. 
Usually, the  methods for validating these operations, if present, are wired in the  object implementation. 
In this work,  
we formalize the notion of a {\em validated object}, decoupling the object operations and properties from the validation procedure. We consider two types of objects, satisfying different levels of consistency: the validated {\em \order} object, offering a total ordering of its operations, and its weaker variant, the validated {\em \regular} object. We provide conditions under which it is possible to implement these objects.
In particular, we show that crash-tolerant implementations of validated \regular objects are always possible in an asynchronous system with a majority of correct processes. However, for validated \order objects, consensus is always required if a property of the object we introduce in this work, {\em persistent validity,} does not hold.
Persistent validity combined with another new property, {\em persistent execution}, allows consensus-free crash-tolerant 
implementations of validated \order objects.
We demonstrate the utility of validated objects by considering several applications conforming to our formalism. 
\end{abstract}

\section{Introduction}

\noindent{\bf Motivation.} 
In distributed computing research, there is an important line of work on the formalization and implementation of distributed concurrent objects. A fundamental challenge of these implementations is making sure the operations that modify the state of an object never drive it into an incorrect or inconsistent state. 
In most proposals, the operations (and their arguments) invoked on the object have been assumed to be always valid, or
ensuring this validity
has been delegated to the application layer.
With the popularization of 
public data structures (due to the wide usage and vast application scope of distributed ledger technologies), 
there is a growing interest on
algorithms and objects capable of tolerating 
non-compliant
user behavior.
In this context, the implementation of an object cannot assume anymore that the operations invoked in the object will be well formed and respect any specification rule. Hence, the implementation of the object must be cautious, and validate operations before applying them. The direct way to do this is to introduce validation tests into the code that implements the object, so that an invalid operation execution is interrupted before it damages the object's state.


In this paper we explore the possibility of separating an object's implementation from the validation of the operations invoked in the object, and the implications of this separation.
This approach is inspired by {\em aspect-oriented programming}~\cite{AOP}, which aims in modular programming by separating cross-cutting concerns, i.e., cohesive areas of functionality. The idea is to add specific checks (advices as called) without changing the code of a program (object in our case).
Our work is meant to be 
a first step on understanding 
how the application requirements and properties impact the algorithms that implement a distributed object through the introduction of a validation predicate $\valid()$ that wraps the application logic of the object.\smallskip

\noindent{\bf Our approach and contributions.} We employ a modular approach in which the characteristics and methods to validate the operations of an object are not ``wired'' in the object implementation. In particular, given a concurrent object $O$ and its supported set of operations $OP$, we recast this object as a {\em validated} object via an $\apply()$ function. This function includes a validation filter, so that a specific operation $op\in OP$ is validated before it is executed,
in accordance to an associated validation predicate $\valid()$. Different validation predicates can be enforced via the $\apply()$ function without affecting the core code of the object. 

\begin{wrapfigure}{R}{0.3\textwidth}
    \begin{minipage}{0.3\textwidth}
    \vspace{-25pt}
\begin{algorithm}[H]
\captionsetup{name=Code}
\small
\begin{algorithmic}[1]
    \State $val \leftarrow \bot$
    \Function{\rread}{~}
        \State \textbf{return} $(val)$
    \EndFunction
    \Function{\wwrite}{$v$}
    \If{$v > 0$} \label{line:test-pos}
        \State $val \leftarrow v$
        \State 
        \textbf{return} $(\ACK)$
    \Else ~
        \textbf{return} $(\NACK)$
    \EndIf
    \EndFunction
\end{algorithmic}
\caption{Implementation of a positive R/W register $O$.}
\label{alg:record-sample-motivation}
\end{algorithm}
   \end{minipage}
  \end{wrapfigure}

Consider the following example. Let $O$ be a simple R/W register supported by two operations, $\rread()$, which returns the value of the register, and $\wwrite(v)$, which changes the value of the register into $v$. Say that 
we would like to impose that only positive numbers are written on the register. One approach would be to include a test directly in the code of the write function (see Code~\ref{alg:record-sample-motivation}). However, should a different or an additional rule be needed, the code would have to be changed again, possibly jeopardizing the implementation correctness (especially in the the case of complex objects). 

\begin{wrapfigure}{R}{0.42\textwidth}
    \begin{minipage}{0.42\textwidth}
    \vspace{-25pt}
\begin{algorithm}[H]
\captionsetup{name=Code}
\small
\begin{algorithmic}[1]
    \Function{\valid}{$op$}
        \State \textbf{return} $(op=\rread() \lor$ 
        \State \hspace{3.5em} $(op=\wwrite(v) \land v>0))$ \label{line:valid-record}
    \EndFunction
    \Function{\execute}{$op$}
    \If{$op=\wwrite(v)$}
        \State $val \leftarrow v$;
        \textbf{return} $(\bot)$
    \Else ~
        \textbf{return} $(val)$
    \EndIf
    \EndFunction
\end{algorithmic}
\caption{Functions \valid and \execute for a positive R/W register $O$.}
\label{alg:record-functions}
\end{algorithm}
   \end{minipage}
  \end{wrapfigure}

With our approach, we separate the test from the code implementing  the object. Processes invoke the desired operation via an $\apply()$ function. The process passes to $\apply()$ the same parameters as it would do in the ``normal'' case, and the apply function invokes a $\valid()$ predicate that has incorporated the desired validation test (i.e., in the case of a write operation, that $v$ is positive, see Code~\ref{alg:record-functions}). In case it is true, it then invokes $\execute()$, which applies the operation on the object (i.e., it sets $v$ as the value of the register). In case the validation fails (e.g., a negative value was intended to be written), $\apply()$ will return a NACK, signaling the violation of the imposed restriction (see Code~\ref{alg:apply-centralized}). Should we require a different validation (e.g., we want a Boolean register), we would only replace the test in the $\valid()$ predicate (e.g., $v >0$ in Line~\ref{line:valid-record} becomes $v \in \{True, False\}$), without making any change on the object's implementation (i.e., in function $\execute()$).

A particular challenge of our approach is to implement the validated version of a given object on a decentralized setting while guaranteeing certain level of consistency. In this work, we consider two types of validated objects, each providing a different level of consistency, the validated {\em \regular } object and the {\em \order} one.
Intuitively, a \regular object provides consistency guarantees similar to a regular register \cite{Lamport86}, while the \order property is similar to linearizability \cite{HW90}.
We are now ready to summarize the contributions of this work. 
\begin{itemize}[leftmargin=3mm]
    \item We introduce the formalization for a generic validated object $O$, along with the two mentioned consistency types, on which the application-specific operations are called (Section~\ref{sec:validated-objects}).
    \item We provide an algorithm to implement validated \regular objects in crash-prone asynchronous distributed systems (Section~\ref{sec:implementing-regular-objects}).
    \item We provide an algorithm to implement validated \order objects under crash or Byzantine failures using the corresponding version of consensus~\cite{Lynch1996, attiya2004distributed} (Section~\ref{sec:consensus}).
    \item Then in Section~\ref{sec:persistent-validity}, we define a property of a validity predicate, which we call {\em persistent validity}, and in Section~\ref{sec:negative} we show that validated \order objects without persistent validity can be used to solve consensus.
    \item In Section~\ref{sec:positivePV}, we introduce an additional property, that we call {\em persistent execution, 
    which allows} a validated \order object to be implemented without consensus. 
    \item Finally, in Section~\ref{sec:applications}, we present some applications (such as a punching system and a cryptocurrency) that conform to the formalism we provided, demonstrating its usability.
    \smallskip
\end{itemize}

\remove{
Then we proceed with the implementation of the validated object according to two consistency  on which we focused, the regular object in Section~\ref{sec:implementing-regular-objects} and the \order object in Section~\ref{sec:implementing-order-objects}. In particular, Section~\ref{sec:persistent-validity} characterize the persistent property of a validity predicate and Section~\ref{sec:negative} shows the implications of the absence of such a property in a generic system. Section~\ref{sec:positivePV} provide the conditions that allow an \order object to be implemented without consensus. Finally Section~\ref{sec:applications} list some applications that conform to the formalism we provided.
}

\noindent{\bf Related work.}
\label{sec:related-works}
The impact of a validation function has been already treated by previous work, according, however, to {\em specific} use cases.

In \cite{frey2021consensus}, a validity property called {\em forward acceptability} is defined, which
enables the operations of a generic application to be commutative. This work only considers eventually consistent objects with this property. It provides an algorithm for a specific case, a PC-Ledger, that is implemented in a consensus-free system. On our side, we have a wider focus, including \textit{any} object, characterizing its validity function and going in detail with the different consistency properties we are able to guarantee.

In~\cite{CachinKPS01}, the authors introduce and solve the notion of Validated Byzantine Agreement to ensure that the decided value is one proposed by a non-faulty process. To do so, they enhance the system with an external validity condition, which requires that the agreement value is valid according to a global, polynomial-time computable predicate, known to all processes and it is application-determined. Hence, each process proposes a value that should satisfy this predicate. Such an external validity condition could be implemented using our approach via an appropriate $\apply()$ function and $\valid()$ predicate (which would implement the required predicate).

In blockchain systems, records are usually validated after the total order  is globally agreed, validating and executing transaction according to the agreed sequence. As an example, Ethereum~\cite{wood2014ethereum} first constructs a block, and then network nodes sequentially run the Ethereum Virtual Machine on each transaction to validate it, and update the global state if valid. This brings to the acceptance and inclusion in the block of invalid transaction inside the global order, in order to gain time in the consensus challenge of the system. A mitigation to this problem is brought by \cite{DBLP:conf/sp/CrainNG21}, that is build on top of \cite{DBLP:conf/nca/CrainGLR18}, where validation is run by a subset of nodes before the proposal is broadcast to the whole network, in order to not overload nodes. In our work we abstract and generalize these behaviors, mapping them to validated objects with different consistency criteria.

In \cite{DBLP:journals/sigact/AntaKGN18}, the authors define a Validated Distributed Ledger Object. The validation is only taken into account in respect to Ledger Objects limiting the scope to that particular kind of data structure. Furthermore, the authors do not investigate or characterize the properties required by validation; they only assume the existence of a validation predicate. 

\section{Validated Objects}
\label{sec:validated-objects}

\subsection{Concurrent Objects and Histories}
\label{sec:object-definition}
We recall the general definition of object formalized in \cite{DBLP:journals/sigact/AntaKGN18} where an object type T specifies $(i)$ the set of values (or states) that any object O of type T can take, and $(ii)$ the set of operations that a process can use to modify or access the value of O. An object O of type T is a concurrent object if it is a shared object accessed by multiple processes \cite{DBLP:books/daglib/0030596,attiya2004distributed}. Each operation on an object O consists of an invocation event and a response event, that must occur in this order. 
A history of operations on O, denoted by $H_O$, is a sequence of invocation and response events, starting with an invocation event. (The sequence order of a history reflects the real time ordering of the events.) An operation $\pi$ is complete in a history $H_O$, if $H_O$ contains both the invocation and the matching response of $\pi$, in this order. A history $H_O$ is complete if it contains only complete operations; otherwise it is partial~\cite{DBLP:books/daglib/0030596}. As in~\cite{DBLP:books/daglib/0030596}, we convert a partial history to a complete one by, for each incomplete operation $\pi$, either removing the invocation of $\pi$ or completing $\pi$ with a response event. From this point onward, we consider only complete histories.
An operation $\pi_1$ precedes an operation $\pi_2$ (or $\pi_2$ succeeds $\pi_1$), denoted by $\pi_1\rightarrow\pi_2$, in $H_O$, if the response event of $\pi_1$ appears before the invocation event of $\pi_2$ in $H_O$. Two operations are concurrent if none precedes the other.
A \emph{run} $R$ of a distributed system that implements object $O$ generates a (potentially infinite) history $H_O$. 

\begin{wrapfigure}{R}{0.4\textwidth}
    \begin{minipage}{0.4\textwidth}
    \vspace{-25pt}
\begin{algorithm}[H]
\captionsetup{name=Code}
\small
\begin{algorithmic}[1]
    \State $S \leftarrow \emptyset$ is the state of the object
    \Function{\apply}{$op, i$}
    \If{$\valid(S, op, i)$}
        \State $r \leftarrow \execute(S,op,i)$
        \State $S \leftarrow S || op$ 
        \State \textbf{return} $(\ACK, r)$
    \Else ~
        \textbf{return} $(\NACK, -)$
    \EndIf
    \EndFunction
\end{algorithmic}
\caption{Centralized implementation of the $\apply$ function for a validated object $O$. Code executed by the central server. Function $\execute(S,op,i)$ provides the result of operation $op$ by process $i$ in state $S$. The operator $||$ combines the new operation with the previous valid executed operations.}
\label{alg:apply-centralized}
\end{algorithm}
   \end{minipage}\vspace{-.3em}
  \end{wrapfigure}

\subsection{Validated Object Types}

In this work we consider validated objects. These are concurrent objects in which the operations executed and how they are interleaved are filtered with a predicate $\valid()$. This predicate has as argument the state of the object $S$ and a new operation $op$ issued by process $i$, and it determines whether $op$ is valid in the light of $S$. The state $S$ is given by an ordered set of operations that have been executed in the object (and are valid). (The operations in $S$ could be concurrent with $op$ but have been ``applied'' in the object before it.) A second function $\execute()$ has the same arguments as $\valid()$ and, if the operation $op$ is valid, is used to obtain the value that $op$ returns.
{We will use the term {\em operation} and the symbol $op$ for custom object logic unknown to our formalism, and we refer to {\em functions} when referring to the primitives of the objects we define, e.g., $\valid$ or $\execute$.}
\remove{
\nn{[NN: that is when executed "after" or concurrently with the operations in $S$?]}\cg{[CG: Since (i) execute is run after valid, and  (ii) S is defined as ``operations that have been executed in the object'', I believe it follows that op is executed after the operations in S, right?]}\af{[In fact the operations in $S$ can be concurrent with $op$, but "take effect" before it.]}
}

In order to use a validated object, a client $i$ 
invokes a function $\apply(op,i)$, which 
checks whether the operation $op$ invoked by $i$ is valid, and if so
it applies it in the object {by executing $op$}. If the operation is not valid, the call $\apply(op,i)$ returns $(\NACK, -)$. If the operation is valid,
$\apply(op,i)$ returns $(\ACK, r)$, where $r$ is the value that $op$ returns. Code~\ref{alg:apply-centralized} shows a
centralized implementation of the function $\apply()$, executed at a single central server. This code is provided for illustration purposes. 
The operator $||$ that defines how the operations are combined into the state $S$ is not detailed on purpose.

\noindent
\fbox{\begin{minipage}{\textwidth} Observe that the history $H_O$ of a run $R$ of a validated object $O$ contains {\em only} the operations $op$ that are found valid and are in fact executed. These are the operations for which $\apply(op,i)$ returns $(\ACK, r)$. We denote the set of complete operations in history $H_O$ generated in run $R$ by $C(R)$.
\end{minipage}}

In the following we assume that $\valid()$ and $\execute()$ have the following arguments:\vspace{-.3em}
\begin{itemize}[leftmargin=7mm]
    \item A strict partially ordered set of operations, given as a pair $\tup{P,\prec}$. $P$ is a set of operations and $\prec$ is a strict partial order defined in~$P$. {In the especial case in which $\prec$ is a total order, denoted as $\pprec$, the first argument can be provided as a
    sequence of operations.}
    \item The operation $op$ to be considered.
    \item The process $i$ that issued $op$.\vspace{-.3em}
\end{itemize}

In this work we consider two types of validated objects.\vspace{-.3em}

\begin{definition}
\label{def:regular-object}
    A validated object specified with functions $\valid()$ and $\execute()$ is a \emph{validated \regular object} if in every run $R$ 
    a partial order $\prec$ among the set $C(R)$ of complete operations can be defined, such that,\vspace{-.5em}
    \begin{enumerate}[leftmargin=5mm]
    \item $\forall op, op' \in C(R), op \rightarrow op' \implies op \prec op'$;
    \item $\forall op \in C(R)$, let $P(op)=\{op': op' \in C(R) \land op' \prec op\}$ and client $i$ the
    issuer of $op$. \sloppy{Then, \mbox{$\valid(\tup{P(op), \prec}, op, i)=True$} and 
    $op$ returns in its response event the value \mbox{$\execute(\tup{P(op), \prec}, op, i)$}.}
    \end{enumerate}
\end{definition}

The following is a 
stronger version  in which 
operations are totally ordered.\vspace{-.5em}

\remove{

\begin{definition}
\label{def:local-object}
    A validated object specified with functions $\valid()$ and $\execute()$ is a \emph{validated \local object} if in every run $R$ a partial order $\prec$ among the set $C(R)$ of complete operations can be defined, such that,\vspace{-1em}
    \begin{enumerate}
    \item $\forall op, op' \in C(R), op \rightarrow op’ \implies op \prec op’$;
    \item $\forall op \in C(R)$, let $P(op)=\{op': op' \in C(R) \land op' \prec op\}$\nn{[NN:do we mean $\prec_l$ for $P(op)$? if not then the definition for valid might not be correct.]} \cg{[CG: Not sure I see the problem here. In $P(op)$ we have the ordering based on $\prec$, but for $\valid()$ and $\execute()$ we require the local order; of course the obtained guarantees are weaker.]} and client $i$ the
    issuer of $op$. Then, $\valid(\tup{P(op), \prec_l}, op, i)=True$ and 
    $op$ returns in its response event the value $\execute(\tup{P(op), \prec_l}, op, i)$.
    \end{enumerate}
\end{definition}
}

\begin{definition}
\label{def:order-object}
A validated object specified with functions $\valid()$ and $\execute()$ is a \emph{validated \order object}\footnote{Note that if $\valid()$ is considered as the sequential specification of the object, then the \order property is a form of linearizability defined over the executed operations. However, to avoid confusion, we prefer to use a different name since we do not include in the object histories the operations that were rejected with $(\NACK, -)$ by 
$\apply()$.} if in every run $R$ a total order $\pprec$ among the set $C(R)$ of complete operations can be defined, such that,\vspace{-.7em}
    \begin{enumerate}[leftmargin=5mm]
    \item $\forall op, op' \in C(R), op \rightarrow op' \implies op \pprec op'$;
    \item $\forall op \in C(R)$, let $P(op)=\{op': op' \in C(R) \land op' 
    {\prec\!\!\!\prec} op\}$ and client $i$ the
    issuer of $op$. \sloppy{Then, \mbox{$\valid(\tup{P(op), \pprec}, op, i)=True$} and
    $op$ returns in its response event the value $\execute(\tup{P(op), {\prec\!\!\!\prec}}, op, i)$.}\vspace{-.5em}
    \end{enumerate}
\end{definition}

In a run $R$ of a validated \order object, the set $C(R)$ is totally ordered by $\pprec$. We will denote the resulting \emph{sequence} of all the operation of $R$ by $S(R)$.\vspace{-.5em}

\section{Algorithms Implementing  Validated Regular Objects\vspace{-.5em}}
\label{sec:implementing-regular-objects}

We present  algorithms implementing validated \regular 
objects
in an asynchronous system.\vspace{-.6em}

\subsection{Model}

We assume a distributed
system composed of $n$ processes with {unique} identities from the set $\idSet=\{1, \ldots, n\}$. Processes are asynchronous and crash prone, i.e.,
they advance their execution at arbitrary speed and can stop permanently (i.e., crash) at any point {during their execution}.
Each process $i$ {has write access to}
a linearizable SWMR Distributed Ledger Object (DLO) \cite{DBLP:journals/sigact/AntaKGN18,DBLP:conf/edcc/CholviAGNR20} denoted $L_i$.
All processes can read all DLOs.
A DLO $L_i$ has a state, which is a totally ordered sequence $S$ of records, initially empty, and has two operations:\vspace{-.8em}
\begin{itemize}[leftmargin=5mm]
    \item $L_i.get()$, which returns the current state (sequence of records) $S$ of the DLO,
    \item $L_i.append(r)$, which adds record $r$ to the end of the sequence~$S$.
\end{itemize}
These DLOs are reliable in the sense that any invocation to these operations by a correct process eventually completes \cite{DBLP:journals/sigact/AntaKGN18}.
Reliable linearizable SWMR DLOs can be implemented in an unreliable asynchronous system. This follows from the work of Imbs et al.~\cite{DBLP:journals/jpdc/ImbsRRS16}, in which they implement SWMR atomic h-registers, which are registers that, when read, return the whole history of written values. Each of these h-registers trivially implements a SWMR DLO. Moreover, their implementation is on a distributed system with $n$ servers, out of which up to $f < n/3$ can be Byzantine. Hence, the implementation is Byzantine-tolerant with optimal resilience.
In this section we also observe that, with minor changes, the algorithm proposed by Imbs et al.~\cite{DBLP:journals/jpdc/ImbsRRS16}  implements SWMR atomic h-registers for $n$ crash-prone processes, out of which $f < n/2$ can fail. 

As usual, we assume the following well-formedness property: A process $i$ does not invoke a call to the function $\apply(op,i)$ of the object being implemented before the previous one has finished.\vspace{-.7em} 
%

\begin{wrapfigure}{R}{0.4\textwidth}
    \begin{minipage}{0.4\textwidth}
    \vspace{-25pt}
\begin{algorithm}[H]
\captionsetup{name=Code}
\small
\begin{algorithmic}[1]
    \Function{apply}{$op,i$}
    \For {$j=1$ to $n$} \label{loop-get}
      \State $G_j  \leftarrow L_j.get()$
      \State $T_j \leftarrow |G_j|$
    \EndFor
    \State $\timestamp \leftarrow (i, T_1, \ldots, T_i, \ldots, T_n)$
    \State $P \leftarrow \{op': \tup{\timestamp',op'} \in \bigcup_j G_j \}$ \label{p-def} 
    \If{$valid(\tup{P, \prec}, op, i)$} \label{p-valid}
        \State $res \leftarrow \execute(\tup{P, \prec}, op, i)$
        \State $L_i.append(\tup{\timestamp,op})$ \label{append-regular}
        \State \textbf{return} $(\ACK, res)$ \label{return-regular}
    \Else ~
        \textbf{return} $(\NACK, -)$
    \EndIf
    \EndFunction
\end{algorithmic}
\caption{Crash-tolerant implementation of the $\apply$ function for a validated \regular object $O$ that uses linearizable SWMR DLOs $L_j, j\in [1,n]$. The code is for process $i \in [1,n]$. }
\label{alg:generic-apply}
\end{algorithm}
    \end{minipage}\vspace{-2.1em}
  \end{wrapfigure}

\subsection{Crash-tolerant Algorithm for Validated Regular Objects}

Let us consider a validated \regular object $O$ specified by the functions $\valid()$ and $\execute()$. Code~\ref{alg:generic-apply} presents an implementation of the $\apply()$ function to be run by each of the processes of the distributed system in order to implement an instance of the object $O$.
For technical reasons, we assume that the invocation action of a valid operation $op$ issued by $i$ occurs when it enters the loop in Line~\ref{loop-get} and invokes $L_1.get()$, and that it completes when the $L_i.append(\tup{\timestamp,op})$ operation in Line~\ref{append-regular} completes. What happens before and after these two actions respectively in the execution of $\apply(op,i)$ is local to process $i$, and not visible outside the client. This assumption removes the uncertainty of whether an operation has been completed if the process crashes after Line~\ref{append-regular} but never executes Line~\ref{return-regular}.

We proceed by demonstrating that Code~\ref{alg:generic-apply} implements a validated \regular object as defined in Definition~\ref{def:regular-object}. 
Observe from Code~\ref{alg:generic-apply} that every operation $op$ issued by process $i$ is assigned a timestamp $\timestamp(op)=(i, T_1, \ldots, T_i, \ldots, T_n)$, which is appended as part of the record of $op$
in ledger $L_i$ if it is valid and completes. The value $T_j$ in the timestamp is the number of records found in ledger $L_j$ in the loop of Line~\ref{loop-get}.
%
%
%
%
These timestamps 
are used to define the partial order $\prec$ among completed operations.

\begin{definition}
\label{def:prec}
Given any two completed operations{} $op, op' \in C(R)$, with respective timestamps $\timestamp(op)=(i,T_1, \ldots, T_i, \ldots, T_n)$ and  $\timestamp(op')=(k,T'_1, \ldots, {T'_i}, \ldots, T'_n)$, $i,k\in [1,n]$,
then $op \prec op'$ if and only if {$T_i < T'_i$}. 
\end{definition}

We show now that $\prec$ is a strict partial order as required.

\begin{lemma}
\label{lemma:strict-order}
If $op \prec op'$ it cannot happen that $op' \prec op$. Hence, $\prec$ is a strict order.
\end{lemma}
 \begin{proof}
 Assume for contradiction that $op \prec op'$ and $op' \prec op$. 
Let $\timestamp(op)=(i,T_1, \ldots, T_n)$ and $\timestamp(op')=(k,T'_1, \ldots, T'_n)$. 
 Then $\apply(op,i)$ finds $T_i$ records in ledger $L_i$ and $T_k$ records in ledger $L_k$, while $\apply(op',k)$ finds $T'_i$ records in $L_i$ and $T'_k$ records in ledger $L_k$. 
 By assumption, we have that $T_i < T'_i$ and $T_k > T'_k$. 
 From $T_i < T'_i$ and the linearizability of $L_i$, the append operation in $\apply(op,i)$ 
 precedes or is concurrent with the $L_i.get()$ operation in $\apply(op',k)$. Hence, $i$ executed $L_k.get()$ before $k$ invoked $L_k.append(\tup{\timestamp(op'),op'})$.
 By the linearizability of $L_k$, it is not possible that $T_k > T'_k$, and we have a contradiction.
 \end{proof}

\begin{lemma}
\label{lemma:implied-order}
$op \rightarrow op' \implies op \prec op'$.
\end{lemma}
\begin{proof}
 Let us assume $op$ was issued by process $i$ and $op'$ was issued by process $k$. Let $\timestamp(op)=(i, T_1, \ldots, T_i, \ldots, T_n)$. From $op \rightarrow op'$, the response action of $op$ happened before the invocation action of $op'$. So, the execution of the append operation $L_i.append(\tup{\timestamp,op})$ in the call $\apply(op,i)$ was completed before the $L_i.get()$ call in  $\apply(op',k)$.
 Then, because of the linearizability of the ledgers, the length of ledger $L_i$ found in $\apply(op',k)$ is $T'_i \geq T_i+1$ (since the append operation increased its length). 
 Hence, $op \prec op'$ from Definition~\ref{def:prec}. \end{proof}

The proof of the next lemma is given in Appendix~\ref{sec:app:proofs}.

\begin{lemma}
\label{lemma:complete-validity}
\sloppy{For each complete operation $op$ (issued by $i$), $\valid(\tup{P(op), \prec}, op, i)=True$. Moreover, $op$ returns in its response event the value $\execute(\tup{P(op), \prec}, op, i)$.}
\end{lemma}


\begin{theorem}
Code~\ref{alg:generic-apply} implements a validated \regular object as defined in Definition~\ref{def:regular-object} in a crash-prone asynchronous system with linearizable SWMR DLOs.
\end{theorem}

\remove{ 
It is not hard to modify Code~\ref{alg:generic-apply} into an algorithm that implements a validated \local object. All there is to do is supply $\valid()$ and $\execute()$ with partial order $\prec_l$ instead of $\prec$.

\begin{corollary}
The modified Code~\ref{alg:generic-apply} implements a validated \local object as defined in Definition~\ref{def:local-object} in a crash-prone asynchronous system with linearizable SWMR DLOs.
\end{corollary}
}

From the fact that reliable linearizable SWMR DLOs can be implemented in a crash-prone asynchronous system \cite{DBLP:journals/jpdc/ImbsRRS16}, we have the following corollary.

\begin{corollary}
It is possible to implement validated \regular 
objects in an asynchronous system with $n$ crash-prone processes from which up to $f < n/2$ can fail. 
\end{corollary}


\section{Validated Totally-ordered Objects}
\label{sec:implementing-order-objects}


\begin{algorithm}[t]
\captionsetup{name=Code}
\small
\begin{multicols}{2}
\begin{algorithmic}[1]
	{\State $S \leftarrow \emptyset$ \Comment{$S$ is a sequence of operations}
    
    \Function{apply}{$op,i$}
    \State $ret \leftarrow \bot$
    \State $AB.broadcast(\tup{op,i})$
    \State wait until $ret \neq \bot$
    \State \textbf{return} $(ret)$
    \EndFunction
    \Statex
    \Upon{$AB.deliver(\tup{op,j})$}
        \If{$\valid(S, op, j)$}
            \State $r \leftarrow \execute(S,op,i)$
            \State $S \leftarrow S || op$ 
            \If{$j=i$}
                $ret \leftarrow (\ACK, r)$
            \EndIf
        \Else
            \If{$j=i$}
                $ret \leftarrow (\NACK, -)$
            \EndIf
        \EndIf
    \EndUpon
    }
\end{algorithmic}
\end{multicols}
\caption{
Implementation of the $apply$ function for a validated \order object $O$ that uses an Atomic Broadcast service. The code is for process $i \in [1,n]$. }
\label{alg:bab-apply}
\end{algorithm}

\subsection{Implementing Validated Totally-ordered Objects with Consensus}
\label{sec:consensus}

We consider now {the set of} validated \order objects as a whole. The first observation is that an object without validation can be seen as a
validated object in which the $\valid()$ predicate always holds. Hence, an object with consensus number $k$~\cite{WaitFree91} will also have at least consensus number 
$k$ in its validated version.

Code~\ref{alg:bab-apply} shows an algorithm that can be used to implement a validated \order object using an Atomic Broadcast service, 
which is known to be equivalent to Consensus~\cite{DBLP:books/sp/Raynal18} (and to MWMR Distributed Ledger Objects~\cite{DBLP:journals/sigact/AntaKGN18}). 
 An Atomic
  Broadcast (AB) service
  \cite{DefagoSU04,coelho2018byzantine,CRISTIAN1995158,DBLP:conf/srds/MilosevicHS11},
       %
       ensures reliable and total ordering of the messages
         exchanged.  Such a communication abstraction is based on
         appropriate {crash-tolerant or} Byzantine-tolerant consensus
        algorithms~\cite{DefagoSU04,DBLP:books/sp/Raynal18}.
       
  \sloppy{The service has two operations, $AB.broadcast(m)$
used by a process to broadcast a message $m$ to all other processes, and $AB.deliver(m)$ used by the service to deliver a message $m$ to a process.}
 From a user point of view, {the AB service} is defined by the following properties:\vspace{-.3em} 
 \begin{itemize}[leftmargin=3mm]
 \item \textit{Validity}:
 if a correct process AB.broadcasts a message, it eventually AB.delivers it.
\item
  \textit{Agreement}: if a correct process AB.delivers a message,
  all correct processes will eventually AB.deliver that message.
\item
  \textit{Integrity}: a message is AB.delivered by a correct process at
  most once, and only if it was previously AB.broadcast.
\item
  \textit{Total Order}: the messages AB.delivered by the correct processes
  are totally ordered (i.e., if a correct process AB.delivers message $m$
  before message $m'$,  every correct process AB.delivers these  message
  in the same order).\vspace{-.3em}
\end{itemize}

Note that if the AB service used is a crash-tolerant one, then Code~\ref{alg:bab-apply} provides crash-tolerant implementation of the $apply$ function, whereas if a Byzantine-tolerant AB service is used, then we have a Byzantine-tolerant implementation of $apply$.
It follows that 
Code~\ref{alg:bab-apply} implements a validated \order object defined by the $\valid()$ and $\execute()$ functions. 
\begin{theorem}
Code~\ref{alg:bab-apply} implements a validated \order object as defined in Definition~\ref{def:order-object} in a fault-prone asynchronous system with an Atomic Broadcast service.\vspace{-.3em}
\end{theorem}

\begin{proof}
The claim holds from the following observations. Firstly, from the Agreement and Total Order properties of the AB service, all correct processes
AB.deliver the same tuples $\tup{op,j}$ in the same order. This guarantees (by induction) that the sequence $S$ maintained in all correct processes is the same. Moreover, for every $op \in S$ it holds that $\valid(S(op), op, j)=True$, where $S(op)$ is the subsequence 
preceding $op$ in $S$. Finally, for every invocation $\apply(op,i)$ by a correct process $i$, the Validity of the AB service guarantees that the tuple $\tup{op,i}$ will be AB.delivered to $i$. Let $S_i(op)$ be the local value of the sequence $S$ when $op$ is AB.delivered to $i$. Then,
the call $\apply(op,i)$ returns $(\NACK, -)$ if $\valid(S_i(op), op, i)=False$, and it returns $(\ACK, \execute(S_i(op), op, i))$ if $\valid(S_i(op), op, i)=True$.
\end{proof}

\remove{
\af{[AF: Give details why this is correct.]}
\cg{[CG: I guess here it should suffice to say that $apply$ is using the total order provided by the AB service to satisfy the properties required in Definition~\ref{def:order-object}.]}
}

\subsection{Persistent Validity}
\label{sec:persistent-validity}

With Code~\ref{alg:bab-apply} we have shown that all validated \order objects can be implemented with a Consensus / Atomic Broadcast service. In this section we explore conditions in the $\valid()$ and $\execute()$ functions that may allow a validated \order object to be implemented without consensus. We first define a property of some objects that we call 
\emph{persistent validity.}

\remove{
\begin{definition}
\label{def:pv}
    Given 
    a validated 
    object together with its associated partial order $\prec$, the $valid()$ function is said to be {\em persistent} iff for every run $R$, 
    $\forall op : \valid(\tup{P(op), \prec}, op, i)=True$ then $\nexists op' \in C(R), j \neq i : \valid(\tup{P(op), \prec}, op', j)=True \land \valid(\tup{P(op) \cup op', \prec}, op, i)=False.$ Recall that $P(op)$ is defined as
    $P(op)=\{op': op' \in C(R) \land op' \prec op\}$, $\forall op \in C(R)$.
\end{definition}
}

\begin{definition}
\label{def:pv}
    Given 
    a validated {\order} object together with 
    its $valid()$ predicate, we say that the object satisfies 
    {\em persistent validity} iff for every run $R$, with order $\pprec$,
    every prefix $S$ of $S(R)$\footnote{Recall that $S(R)$ is the sequence of operations in $C(R)$ totally ordered by $\pprec$.},
    and every operation $op_i \notin S$,
        if $\valid(S, op_i, i)=True$ then 
    $\nexists op_j \not\in S, j \neq i:
     \valid(S, op_j, j)=True \land \valid(S || op_j, op_i, i)=False.$
\end{definition}

Persistent validity informally says that once an operation is valid, then it cannot be made invalid by operations issued by the other processes.
%
In Section~\ref{sec:positivePV}
we show how {\em persistent validity} can help in the implementation of different objects according to different consistency criteria. But first, we investigate the implications of a validation predicate $\valid()$ for which the {\em persistent validity} does not hold.

\remove{
We will be referring to the above as the {\em persistent validity (PV) property.}

In respect to Definitions~\ref{def:linearizable-object} of a validated \linearizable object, the {\em persistent validity} property definition remains unchanged except for the $\pprec$ operator in the calls to $\valid()$ that changes to $\pprec_l$.
}

\subsection{Total Order Without Persistent Validity Is as Strong as Consensus}
\label{sec:negative}

We demonstrate that a validated \order object whose $\valid()$ function does not satisfy the {\em persistent validity} property, is as strong as consensus. In order to do that, we will demonstrate that such an object can be used to solve the consensus problem between two crash-prone processes in an asynchronous system with at most one failure. 
%

\begin{observation}
\label{obs:non-pv-linearizable}
Let $O$ be a validated \order object \emph{without persistent validity.} Then, there is a run $R$ of $O$, a prefix $S \subseteq S(R)$, and operations $op_i, op_j \notin S$ issued by processes $i \neq j$, such that $\valid(S, op_i, i)=True
\land \valid(S, op_j, j)=True \land \valid(S || op_j, op_i, i)=False.$
\end{observation}
Informally, Observation~\ref{obs:non-pv-linearizable} says that there is a run $R'$ derived from $R$ in which $op_i$ issued by client $i$ is valid if ordered after $S$, but there exists another valid operation $op_j$ issued by a client $j \neq i$ that, if ordered before $op_i$, invalidates it.
Note that no information is given on $op_j$, so it is not known if the inverse is true, i.e., whether $op_i$, if ordered before $op_j$, invalidates it.


We show now that object $O$, the prefix $S$, and the operations $op_i$ and $op_j$ can be used by processes $i$ and $j$ to reach consensus in an asynchronous system in which one of them can fail by crashing. Since without the object $O$ it is known that in such a system consensus cannot be solved, we conclude that $O$ is what allows to solve consensus.

In the rest of the section we hence assume a distributed system in which (at most) one process can crash, and computations happen in an asynchronous way so we can not make any assumption about processes relative speeds. The object $O$ is assumed to be reliable, i.e. it does not fail or crash in any way. In addition, processes $i$ and $j$, and the object $O$ can use a reliable shared memory formed of atomic SWMR registers. As said before, such a shared memory can be implemented in an asynchronous message passing system if a majority of processes is correct
\cite{DBLP:journals/jacm/AttiyaBD95,DBLP:journals/jpdc/ImbsRRS16}. For our results to hold it is enough to assume that at most one process can crash, hence we assume $f=1$.
Then, while we focus on achieving consensus between processes $i$ and $j$, if required, in order to implement the object $O$ and the shared memory other processes can be involved. (In particular, at least a third process participates in the implementation of the shared memory to fulfill the requirement of a majority of correct processes.)

\begin{algorithm}[t]
\setlength{\columnsep}{3em}
\captionsetup{name=Code}
\small
\begin{algorithmic}[1]
\begin{multicols}{2}
\State Initialize object $O$ with the prefix $S$
\State Init: $cons\_register_i$ and $cons\_register_j$ are atomic SWMR registers writable only by $i$ and $j$ respectively, initially $\bot$.

\State Code for process $i$:
\Function{propose}{$v_i$} 
    \State \textbf{write}($cons\_register_i, v_i$)
    \State $r \leftarrow O.apply(op_i, i)$ 
    \If{$r = (\NACK, -)$}
        \State $v_j \leftarrow$ \textbf{read}($cons\_register_j$)
        \State \textbf{decide}($v_j$)
    \Else ~
        \textbf{decide}($v_i$)
    \EndIf
\EndFunction
\State Code for process $j$:
\Function{propose}{$v_j$}
    \State \textbf{write}($cons\_register_j, v_j$)
    \State $r \leftarrow O.apply(op_j, j)$
    \If{$r = (\NACK, -)$}
        \State $v_i \leftarrow$ \textbf{read}($cons\_register_i$)
        \State \textbf{decide}($v_i$)
    \Else ~
        \textbf{decide}($v_j$)
    \EndIf
\EndFunction
\end{multicols}
\end{algorithmic}
\caption{Algorithm solving consensus for processes $i$ and $j$ when $op_i$ and $op_j$ invalidate each other.}
\label{alg:impossibility-1}\vspace{-1em}
\end{algorithm}

\noindent
\textbf{Both operations invalidate each other.}
Let us first assume that the two operations $op_i$ and $op_j$ are exclusive, i.e., if any one of the two is executed on the object $O$ after prefix $S$, it makes invalid the other. In this case, Code~\ref{alg:impossibility-1} can be used by processes $i$ and $j$ to reach consensus. Observe that the code used by the two processes is completely symmetric.

\begin{lemma}
\label{lem:inval}
Let $O$ be a validated \order object without persistent validity, and let 
prefix $S$, processes $i$ and $j$, and operations $op_i, op_j \notin S$ as in Observation~\ref{obs:non-pv-linearizable}. Moreover, assume that
$\valid(S || op_i,$ $op_j, j)=False$. Then Code~\ref{alg:impossibility-1} allows 
processes $i$ and $j$ to reach consensus. 
\end{lemma}
\begin{proof}
     Process $i$ first writes its proposed value $v_i$ in its own register $cons\_register_i$ and then calls $O.apply(op_i,i)$. Process $j$ does the same with register $cons\_register_j$ and call $O.apply(op_j,j)$. By assumption, only one of the operations $op_i$ and $op_j$ is found valid. Then, if process $i$ receives an $\ACK$ from $apply(op_i,i)$, it can safely decide $v_i$, knowing that process $j$ will receive $\NACK$ and decide $v_i$ as well. On the other hand, if process $i$ receives $\NACK$ and process $j$ receives $\ACK$, value $v_j$ is decided by both processes.
     
     Let us now assume that one process crashes; wlog, process $j$. If process $j$ never issued the call $O.apply(op_j,j)$ or the call was issued but $op_j$ was found invalid, then $O.apply(op_i,i)$ will return $\ACK$ and process $i$ will decide
     $v_i$. If, on the other hand, $j$ issued the call $O.apply(op_j,j)$ and $op_j$ was found valid, then process $i$ receives a $\NACK$, and reads $cons\_register_j$. Since the value $v_j$ was written in $cons\_register_j$ by $j$ before calling $O.apply(op_j,j)$, the read operation completes and returns $v_j$, which is the value decided by process $i$. Process $j$ cannot decide a different value, since 
     $O.apply(op_j,j)$ returns $\ACK$.\vspace{-.5em}
\end{proof}

\begin{wrapfigure}{R}{0.47\textwidth}
    \begin{minipage}{0.47\textwidth}
    \vspace{-25pt}
\begin{algorithm}[H]
\captionsetup{name=Code}
\small
\caption{$\mathit{\ApplyIfValid}(op)$ function to communicate with $O$. 
It returns $(s,r)$, where $s \in \{\ACK, \NACK\}$. 
Code for process $k$.}
\label{alg:applyIfValid}
\begin{algorithmic}[1]
    \State init: $oplist_k$ are SWMR vectors writable only by $k$, initially $\bot$
    \State init: $reslist_k$ are SWMR vectors writable only by $O$, initially $\bot$
    \State init: $c_k \leftarrow 1$ \Comment{$c_k$  is a local variable of $k$}
    \Function{$\ApplyIfValid$}{$op,k$}
        \State \wwrite($oplist_k[c_k], op, k$)
        \State wait until $reslist_k[c_k] \neq \bot$
        \State $res \leftarrow$ \rread($reslist_k[c_k]$)
        \State $c_k \leftarrow c_k + 1$
        \State \textbf{return} $(res)$
    \EndFunction
\end{algorithmic}
\end{algorithm}
    \end{minipage}\vspace{-.8em}
  \end{wrapfigure}

\remove{
In the general case, we only know by hypothesis that one operation gets invalidated and we don't know nothing about the validity of the other one. 
We now deal with the case in which
the second operation does not get invalidated by the other one. For simplicity, $op_i$ is the operation that gets invalidated by $op_j$ but $op_j$ remains valid independently from $op_i$. \af{[AF: I do not agree this to be the general case. The general case is when $op_i$ is the operation that gets invalidated by $op_j$, and $op_j$ may or may not be invalidated by $op_i$. If we want to keep them as two different cases is fine, but we have to be careful because if we remove the case in which they invalidate each other we have to use this version I describe.]}\ar{[AR: We can make it general adding a check on $res$ in Line 16 of Code~\ref{alg:impossibility-2}]}
}

\noindent
\textbf{Operation $op_i$ does not invalidate operation $op_j$.}
We now deal with the case in which $op_j$ makes $op_i$ invalid, but $op_i$ does not invalidate $op_j$. Observe that Code~\ref{alg:impossibility-1} does not solve this case because, since $op_j$ is always valid, the value returned by call $O.apply(op_j,j)$
does not allow process $j$ to know whether $op_i$ was found valid.
%
Notice that we use the validated \order object as a black box. Therefore, process $j$
does not have direct access to the totally-ordered sequence of operations in the object. 
Thus, for process $j$ to know whether $op_i$ is found valid some
extra work needs to be done.
The key of the solution is the use of the shared memory available in the system to log
the calls $O.\apply()$ and the values they return. To do so, a generic process $k$ has a
SWMR vector $oplist_k$ through with $\apply()$ call are issued. The result of the call is written by object $O$ in a SWMR vector $reslist_k$ from where $k$ can read it. This process is encapsulated in the side of the generic caller process $k$ in the function
$\ApplyIfValid()$ presented in Code~\ref{alg:applyIfValid}.

\begin{wrapfigure}{R}{0.46\textwidth}
    \begin{minipage}{0.46\textwidth}
    \vspace{-25pt}
\begin{algorithm}[H]
\captionsetup{name=Code}
\small
\caption{Task executed by object $O$ to process the $\apply()$ calls issued by process $k$.}
\label{alg:applyJournaled}
\begin{algorithmic}[1]
    \State Init: $oplist_k$ and $reslist_k$ are the vectors from Code~\ref{alg:applyIfValid}
    \State init: $c_k \leftarrow 1$ \Comment{$c_k$  is a local variable of $O$}
    \Loop
        \State wait until $oplist_k[c_k] \neq \bot$
        \State $op \leftarrow$ \rread($oplist_k[c_k]$)
        \State $res \leftarrow \apply(op,k)$
        \State \wwrite($reslist_k[c_k], res$)
        \State $c_k \leftarrow c_k + 1$
    \EndLoop
\end{algorithmic}
\end{algorithm}
 \end{minipage}\vspace{-1.5em}
  \end{wrapfigure}

On its side, object $O$ is waiting for $\apply()$ calls being issued via the vector
$oplist_k$, and when one appears it applies it and writes in $reslist_k$ the corresponding result. This can be implemented with one concurrent task for each process $k$ as presented in Code~\ref{alg:applyJournaled}. Note that, since the object $O$ and the shared memory are both reliable, if an $\apply()$ call is written by process $k$ in $oplist_k$, eventually the corresponding response will be written in $reslist_k$, even if $k$ has crashed in the mid time.

With this logged method of using the object, the algorithm that processes $i$ and $j$ can use to solve consensus is presented in Code~\ref{alg:impossibility-2}. Observe that the
code for process $i$ is similar to the one in Code~\ref{alg:impossibility-1}, replacing
the call $O.\apply(op_i,i)$ with call $\ApplyIfValid(op_i,i)$. However, the code for process $j$ is different, since it has to access $oplist_i$ and $reslist_i$ to determine whether $op_i$ was found valid. 

\remove{
In order to proceed with the proof, we define in Code~\ref{alg:applyJournaled} the behavior of the $apply$ function provided by the validated \cg{\order} 
object. Code~\ref{alg:applyIfValid} is executed by the client according to Code~\ref{alg:impossibility-2}. In practice two vectors of registers, $oplist$ and $reslist$ keep track of the interaction among clients and the object. The code of client $i$ differs from the code of client $j$ to cover the asymmetry of behavior between $op_i$ and $op_j$.
\cg{[CG: I would make it a theorem.]}
\begin{lemma}
    Given a validated \cg{\order} 
    object whose $valid$ function conforms to Definition~\ref{def:non-pv-linearizable} then Code~\ref{alg:applyJournaled},  \ref{alg:applyIfValid}, and \ref{alg:impossibility-2} allow two arbitrary processes to reach consensus. 
\end{lemma}
\cg{[CG: In the proof below, we need to also consider the case of one of the processes failing - related to a previous comment.]} \af{[AF: Yes, the case of $i$ or $j$ failing has to be considered. Note that the object is assumed to be reliable, so when the operation is written in the shared memory there will always be a response from the object.}
}

\begin{lemma}
\label{lem:no-inval}
Let $O$ be a validated \order object without persistent validity, and let 
prefix $S$, processes $i$ and $j$, and operations $op_i, op_j \notin S$ as in Observation~\ref{obs:non-pv-linearizable}. Moreover, assume that
$\valid(S || op_i,$ $op_j, j)=True$. Then Codes~\ref{alg:applyJournaled},  \ref{alg:applyIfValid}, and \ref{alg:impossibility-2} allow 
processes $i$ and $j$ to reach consensus. 
\end{lemma}

 \begin{proof}
     \sloppy{Without crashes, both processes $i$ and $j$ start by writing their proposed values $v_i$ and $v_j$ in their respective $cons\_register_i$ and $cons\_register_j$.} Then, they call $\ApplyIfValid()$ with their operations $op_i$ and $op_j$. As in Code~\ref{alg:impossibility-2}, process $i$ waits for response and decides $v_i$ or $v_j$ depending on whether $op_i$ was found valid or not. This is determined from the value returned by the $\ApplyIfValid(op_i,i)$ call.
    
     On its hand, process $j$ always receives $\ACK$ from the $\ApplyIfValid(op_j,j)$ call,
     since operation $op_j$ is found valid by hypothesis. So, it can not use this to know whether $op_i$ precedes $op_j$ and was hence found valid. Instead,
    it first checks if process $i$ submitted $op_i$ via a $\ApplyIfValid(op_i,i)$ call by searching in the $oplist_i$ vector. If $op_i$ was not submitted, then process $j$ can safely decide $v_2$ (line~\ref{line:imp2:op_sub_check_no}), because if it is
    submitted now it will be found invalid. Note that process $i$ will decide $v_2$ as well.
    
     If $op_i$ is found in $oplist_i$ (line~\ref{line:imp2:op_sub_check}), then process $j$ needs to wait for the result of $\ApplyIfValid(op_i,i)$ by reading from register $reslist_i$. As mentioned, because of the reliability of the object and the shared memory, the result will eventually be written there. At this point, if the result of $\ApplyIfValid(op_i,i)$ is $\ACK$ then it means that $op_i$ was ordered before $op_j$, and the value to be decided is $v_i$. If it is $\NACK$ then $op_j$ has been ordered
    before $op_i$, $op_i$ was invalid, and the value to be decided is $v_j$. In either case, the decided value is consistent with the one decided by process $i$, solving consensus between the two processes.
    
    The correctness for the case when process $j$ crashes is as in the proof of Lemma~\ref{lem:inval}. If process $i$ crashes before writing $op_i$ in $oplist_i$, then $j$ decide $v_j$ as described above. Otherwise, $op_i$ will be processed by $O$ and found valid (and $v_1$ will be the decided value in both processes) or invalid (and $v_2$ will be the decided value).
 \end{proof}
 
 \begin{algorithm}[t]
\setlength{\columnsep}{3em}
\captionsetup{name=Code}
\small
\begin{algorithmic}[1]
	\begin{multicols}{2}
\State Initialize object $O$ with prefix $S$
\State Init: $cons\_register_i$ and
$cons\_register_j$ are SWMR registers writable only by $i$ and $j$ respectively, initially $\bot$
\State Init: $oplist_i$ and $reslist_i$ are the vectors from Code~\ref{alg:applyIfValid}

        \State Code for process $i$:
        \Function{propose}{$v_1$}
            \State \wwrite($cons\_register_i, v_1$)
            \State $res \leftarrow \ApplyIfValid(op_i,i)$ 
            \If{$res = (\NACK, -)$}
                \State $v_2 \leftarrow$ \rread($cons\_register_j$)
                \State \textbf{decide}($v_2$)
            \Else ~
                \textbf{decide}($v_1$)
            \EndIf
        \EndFunction
        \vfill\null
        \columnbreak
    
        \State Code for process $j$:
        \Function{propose}{$v_2$}
            \State \wwrite($cons\_register_j, v_2$)
            \State $res \leftarrow \ApplyIfValid(op_j,j)$
            \If{$\exists c: op_i = oplist_i[c]$} \label{line:imp2:op_sub_check}
                \State wait until $reslist_i[c] \neq \bot$
                \State $opires \leftarrow$ \rread($reslist_i[c]$)
                \If{$opires = (\ACK, r)$}
                    \State $v_1 \leftarrow$ \rread($cons\_register_i$)
                    \State \textbf{decide}($v_1$)
                \EndIf
                \If{$opires = (\NACK,-) $}
                    \State \textbf{decide}($v_2$)
                \EndIf
            \Else ~ \label{line:imp2:op_sub_check_no}
                \textbf{decide}($v_2$) \label{line:imp2:j_decide_v2}
            \EndIf
        \EndFunction
\end{multicols}
\end{algorithmic}
\caption{Algorithm that solves consensus for processes $i$ and $j$ when $op_i$ does not invalidate $op_j$.}\vspace{-2em}
\label{alg:impossibility-2}
\end{algorithm}

\begin{definition}
\label{def:pe}
    Given a validated {\order} object together with its associated 
    $\valid()$ predicate and $\execute()$ function,
    we say that the object satisfies \emph{persistent execution} iff 
    \begin{enumerate}
        \item it satisfies persistent validity and
        \item \sloppy{for every run $R$, with order $\pprec$, every prefix $S$ of $S(R)$, and every pair of operations $op_i, op_j \notin S$ from processes $i \neq j$, if $\valid(S, op_i, i)=True \land \valid(S, op_j, j)=True$ then $\execute(S, op_i, i) = \execute(S || op_j, op_i, i).$}
    \end{enumerate}
\end{definition}


\begin{theorem}
\label{thm:consensus-to}
Let $O$ be a validated \order object without persistent validity, then $O$ can be used
to solve consensus in a crash-prone asynchronous system with $n \geq 3$ processes in which at most one process can crash.
\end{theorem}
\begin{proof}
From Lemmas~\ref{lem:inval} and \ref{lem:no-inval} we have that two processes $i$ and $j$ can solve consensus between them. To make the solution applicable to the $n$ processes, and allow any of the $n$ values proposed to be decided, we have each process writing its proposed value in a SWMR register $\mathit{prop}_k$ in the shared memory. Processes $i$ and $j$ wait until $n-1$ such registers are filled, and choose one value from these values. Then, they run the consensus algorithm between them, proposing the chosen values. As soon as one of the two processes decides, it writes the decision in the shared memory. They use SWMR registers $\mathit{decision}_i$ and $\mathit{decision}_j$. Since at least one process $i$ or $j$ is correct, the value decided is eventually written in at least one of these registers.
Then, the other processes can read it from there and also decide.
\end{proof}

Observe that in a crash-prone asynchronous system in which one process can crash consensus cannot be solved \cite{FLP85}. Thus, Theorem~\ref{thm:consensus-to} implies that
any validated \order object without persistent validity is as strong as a consensus object \cite{DBLP:books/daglib/0030596} in such a system.

\subsection{Consensus-free Total Order with Persistent Execution}
\label{sec:positivePV}
The previous result shows that persistent validity is required in order to be able to implement a validated {\order} 
object without consensus. Unfortunately this is not enough, as can be trivially observed from the fact that a $\valid()$ predicate that is always $True$ satisfies persistent validity. To be able to implement the object without consensus, some \emph{additional} condition must be imposed. The following is {an instance of} such a condition.

An object with persistent execution has significant flexibility for reordering concurrent operations to obtain different total orders $\pprec$ that satisfy the conditions of Definition~\ref{def:order-object}. Consider a validated \order object $O$ and a finite run $R$. Let $K$ be a set of concurrent operations issued by different processes $k$, such that $\forall op_k \in K$, it holds that $op_k \notin S(R)$,
$\nexists op \in S(R): op_k \rightarrow op$, and $\valid(S(R), op_k, k)=True$.

\begin{lemma}
If the validated \order object $O$ satisfies persistent execution, then $R$ 
can be extended with all the operations in $K$, in any order, satisfying Definition~\ref{def:order-object}. Moreover, $\forall op_k \in K$,
the value returned by $O.\apply(op_k,k)$ is $(\ACK,\execute(S(R), op_k, k))$.
\end{lemma}
\begin{proof}
Any extension $R'$ as described will respect property (1) of Definition~\ref{def:order-object}, because the operations in $K$ do not precede in real time order those in $S(R)$ and they are concurrent among themselves.
Regarding property (2), from the assumption that $\forall op_k \in K, \valid(S(R), op_k, k)=True$, that the operations in $K$ are issued by different processes, and persistent validity, we have that all operations in $K$ will be valid in the extension of $R$.
Finally, the value returned for $op_k$ will be $(\ACK,\execute(S(R), op_k, k))$ from
property (2) of Definition~\ref{def:pe}.\vspace{-1em}
\end{proof}

From this lemma, we can derive that Code~\ref{alg:generic-apply} implements a validated \order object $O$ when persistent execution is satisfied. 
The total order $\pprec$ of a run of $O$ has to be an extension of the
order from Def.~\ref{def:prec}, imposing an order among those operations that are not ordered by $\prec$. One possibility is to order complete operations in a run by the real time order of their response events in the history of the run. This total order is consistent with $\prec$ because Def.~\ref{def:prec} and Code~\ref{alg:generic-apply} guarantee that (1) if $op \prec op'$ then $op$ completes before $op'$ and that (2) if $op$ completes before $op'$ and $op \not\prec op'$ then
$op$ and $op'$ are concurrent.
Hence the following result, which implies that consensus is not required to
implement validated \order objects with persistent execution.

\begin{theorem}
It is possible to implement a validated \order object $O$ that satisfies persistent validity and persistent execution in an asynchronous system with $n$ crash-prone processes from which up to $f < n/2$ can fail.\vspace{-.5em}
\end{theorem}

\section{Applications of Validated Objects}
\label{sec:applications}
To demonstrate the usefulness of validated 
objects, in this section we present a 
number of possible applications providing 
the exact properties that each application 
satisfies. For each application, we present both a {\em relaxed} version, i.e, one that uses regular validated objects, and a {\em strict} version, i.e, one that uses totally-ordered validated objects,
and we analyze what
validity properties are required for the applications being realized. (One more application is given in Appendix~\ref{sec:Versioned}.)\vspace{-.5em}

\subsection{Punching System}
A punching system is an object that can be used by a process to log its activity. It essentially allows a process to signal the start of an activity and then signal that activity's end. 
One practical such system is used for tracking 
employee arrival and departure in various organisations. Such object may have the following two operations:\vspace{-.5em}
\begin{itemize}[leftmargin=5mm]
    \item $\act{punch-in}(t, i)$, that can only be invoked by process $i$, to mark his arrival at time $t$,
    \item $\act{punch-out}(i)$, that can only be invoked by process $i$, to mark his departure and return the hours worked since he last punched-in
\end{itemize}
Notice that the $\act{punch-in}(t,i)$ operation 
for $i$ is only valid if the last operation 
from $i$ was a $\act{punch-out}(i)$ operation and vice-versa.

\begin{wrapfigure}{R}{0.55\textwidth}
    \begin{minipage}{0.55\textwidth}
    \vspace{-37pt}
\begin{algorithm}[H]
\captionsetup{name=Code}
\caption{Functions $\valid()$ and $\execute()$ to implement a punching system object.}
\label{alg:punch}
\begin{algorithmic}[1]
\Function{$\valid$}{$\tup{P,\prec}, op, i$}
   \If{$(i$ is not the issuer of $op)$}
       \State \textbf{return}($False$)
    \EndIf
    \State $lop_i \leftarrow \{op': op'$ the last operation of $i$ in $P\}$
    \If{$(op=\act{punch-out}(k))$}
            \State \textbf{return} ($i = k \land 
            lop_i = \{\act{punch-in}(t,i)\} $)
    \Else \Comment{$op= \act{punch-in}(t, k)$}
            \State \textbf{return} ($i = k \land op = \act{punch-in}(t, k) \land$ \hfill \break
            \hspace*{7em}$( lop_i=\emptyset \lor lop_i = \{\act{punch-out}(i)\} ) $)
    \EndIf
\EndFunction\vspace{-1em}

\Statex

\Function{$\execute$}{$\tup{P,\prec}, op, i$}
    \If{$(op=\act{punch-out}(i))$}
        \State $lt_i \leftarrow \{t: op'=\act{punch-in}(t,i) \land$ \hfill \break
        \hspace*{6em} $\nexists op'' \in P$  s.t. $op'\prec op'' \}$
        \State \textbf{return} ($\mathit{hours}(\mathit{now}()-lt_i)$)
    \Else ~
        \textbf{return} ($\bot$)
    \EndIf
\EndFunction
\end{algorithmic}
\end{algorithm}
\end{minipage}\vspace{-1.2em}
\end{wrapfigure}

This object has both the persistent validity and persistent execution properties, as whenever 
$i$ recorded a $\act{punch-in}$ operation the $\act{punch-out}$ operation remains valid no matter of the operations executed by any other
process $j$. Persistent execution also holds 
since the value of the object at $i$ remains the same until $i$ performs its  $\act{punch-in}$ or
$\act{punch-out}$ operations. 

Notice that since the process $i$ is 
restricted to obtain its own working 
hours (i.e., invoke only $\act{punch-out}(i)$) then by well-formedness the relaxed version of the punching system is equivalent with the strict version. Thus, the system may be 
implemented without consensus utilizing 
the functions defined in Code~\ref{alg:punch}.
Recall in this code that $P$ is the set of complete operations that precede $op$ using the order $\prec$. Note also that
$\prec$ orders all the operations from the same process (from well-formedness and property (1) of Definitions~\ref{def:regular-object} and \ref{def:order-object}), so $lop_i$ is well defined.

\begin{theorem}
Code~\ref{alg:punch} combined with Code~\ref{alg:generic-apply} implements a strict punching system that satisfies both 
persistent validity and persistent execution.  
\end{theorem}

\subsection{Cryptocurrency}

\sloppy{In this section we implement a cryptocurrency (asset transfer)~\cite{DBLP:conf/podc/GuerraouiKMPS19}.} For that, a validated 
object is created, which holds an account for each process in $[1,n]$. For simplicity we assume that each process has initially a balance of $\ibalance$ tokens. The object has only two operations as described in~\cite{DBLP:conf/podc/GuerraouiKMPS19}:\vspace{-.5em}
\begin{itemize}[leftmargin=5mm]
    \item $\transfer(i, k, x)$, that can only be invoked by process $i$, transfers $x>0$ tokens from the account of the issuing process $i$ to the account of process $k$, and
    \item $\balancex(k)$, which returns an estimate of the current balance of 
    the account of process $k$. 
\end{itemize}
We assume that the operations are cryptographically signed by the issuer. As usual, it is not allowed that a process ever has negative balance. Hence, an operation $\transfer(i, k, x)$ 
is valid and can be executed only when the
balance of $i$ 
is higher than the amount $x$ to be transferred. In~\cite{DBLP:conf/podc/GuerraouiKMPS19} this validation is embedded of the operation execution, while here validation and execution are separated in different functions $\valid()$ and $\execute()$ (see Code~\ref{alg:cryptocurrency}).

\begin{wrapfigure}{R}{0.6\textwidth}
    \begin{minipage}{0.6\textwidth}
    \vspace{-25pt}
\begin{algorithm}[H]
\captionsetup{name=Code}
\caption{Functions $\valid()$ and $\execute()$ to implement a cryptocurrency.}
\label{alg:cryptocurrency}
\begin{algorithmic}[1]
\Function{$\valid$}{$\tup{P,\prec}, op, i$}
    \If{$(i$ is not the issuer of $op) \lor$ \hfill \break
    \hspace*{2.5em} $($signature of $op$ is invalid$)$}
       \textbf{return} ($False$)
    \EndIf
    \If{$op=\balancex(k)$}
        \textbf{return}($True$)
    \Else \Comment{$op= \transfer(j, k, x)$}
        \If{$(op= \transfer(j, k, x) \land j\neq i) \lor (x \leq 0)$}
            \State \textbf{return} ($False$)
        \EndIf
        \State $b_{in} \leftarrow \ibalance + \sum \{x': \exists j, \transfer(j,i, x') \in P  \}$
        \State $b_{out} \leftarrow \sum \{x': \exists j, \transfer(i,j,x') \in P \}$
        \State\textbf{return} ($b_{in} - b_{out} \geq x$)
    \EndIf
\EndFunction\vspace{.2em}
\Function{$\execute$}{$\tup{P,\prec}, op, i$}
    \If{$op=\balancex(k)$}
        \State $b_{in} \leftarrow \ibalance + \sum \{x': \exists j, \transfer(j,k,x') \in P  \}$
        \State $b_{out} \leftarrow \sum \{x': \exists j, \transfer(k,j,x') \in P \}$
        \State \textbf{return} ($b_{in} - b_{out}$)
    \Else ~ \textbf{return} ($\bot$) \Comment{$op= \transfer(i,k,x)$}
    \EndIf
\EndFunction
\end{algorithmic}
\end{algorithm}
\end{minipage}\vspace{-1.2em}
\end{wrapfigure}

We can also get this object in two flavors. 
In the relaxed version of the object the value returned by the $\balancex(k)$ operation must include all operations that precede the $\balancex(k)$ in real time ordering, but may not include some of the $\transfer()$ operations that are concurrent with the call. Thus, the operation may return 
a lower bound of the actual balance (including the concurrent operations). On the other hand, 
in the strict version (i.e., if we use a validated totally-order object) 
then the balance operations will return the 
exact amount of the balance. 
The same applies to $\transfer()$ operations. In the relaxed version some of them may be found invalid because incoming funds in concurrent transfers are not accounted for.
%

In order to implement the relaxed version of this object, it is enough to use the functions $\valid()$ and $\execute()$ as defined in Code~\ref{alg:cryptocurrency}, and use them in Code~\ref{alg:generic-apply}. Observe that 
the cryptocurrency object satisfies the property
of persistent validity but it does not satisfy 
the property of a persistent execution.
Therefore, in order to implement the strict version of this object, 
one may combine the functions of Code~\ref{alg:cryptocurrency}
with  Code~\ref{alg:bab-apply}, which uses 
the Atomic Broadcast service.\vspace{-.5em} 

\begin{theorem}
Code~\ref{alg:cryptocurrency} combined with Code~\ref{alg:generic-apply} or with Code~\ref{alg:bab-apply}, implement 
the relaxed and strict cryptocurrrency (asset transfer), respectively.\vspace{-.5em} 
\end{theorem}




\remove{
\subsection{Voting System without Consensus}

Consider now a voting system in which all participants can cast a yes/no vote at any time, only one per participant. The coordinator has an operation to query the object about whether a decision has been reached. This system has persistent validity but not persistent execution. It can be implemented as a validated \regular object. \af{Can be linearizable?}

With PV: Operations that start voting and vote. Voting several times is not allowed. The voting ends when everyone voted.
} 

\subsection{Do-All: Task Execution}

Do-All is an object in which a set of processes execute tasks taken from a set of jobs~\cite{DoAll1,DoAll2}. Any process can take
any task from the job set. In respect to the number of jobs that processes should execute, the specification of the Do-All can be traced to (1) a strict validated totally-ordered object, Definition~\ref{def:order-object}, if a specific number $T$ of job's executions must be respected, or (2) to a relaxed validated regular object, Definition~\ref{def:regular-object}, if executions of jobs beyond the threshold $T$ can be tolerated when some conditions are met, e.g., if they were initiated in a batch of concurrent operations.


Specifically, the Do-All object supports the following operations:\vspace{-.7em}
\begin{itemize}[leftmargin=10mm]
\item $Do (t,i)$: process $i$ claims and performs task $t$.
\item $Completed(t, i)$: process $i$ reports the completion of task $t$.
\end{itemize}

$Do()$ is not valid if a certain number of processes performed the task (say 3 for redundancy). Notice that, as mentioned,  
we can have the strict version (i.e., totally-order version) of the object, in which \emph{exactly} 3 processes can do a task, and the relaxed version (regular version) in which {\em at least 3} do it. 
Observe that this object does not satisfy the persistent validity property, nor the persistent execution one.
Code \ref{alg:doall} shows an implementation for the $\valid$ and $\execute$ predicates to realize the Do-All object in both cases.
The following result holds.

\begin{theorem}
Code~\ref{alg:doall} combined with Code~\ref{alg:generic-apply} or with Code~\ref{alg:bab-apply}, implement 
the relaxed and strict Do-All object, respectively. \vspace{-1em}
\end{theorem}

\begin{algorithm}[t]
\captionsetup{name=Code}
\begin{multicols}{2}
\begin{algorithmic}[1]
{\Function{$\valid$}{$\tup{P,\prec}, op, i$}
   \If{$(i$ is not the issuer of $op)$}
       \State \textbf{return}($False$)
    \EndIf
    \If{$(op=\act{completed}(x, k))$}
        \State \textbf{return}($i = k$)
    \Else \Comment{$op= \act{do}(x, k)$}
        \If{$(op= \act{do}(x, k) \land i = k)$}
            \State $c \leftarrow |\{j: \act{do}(x,j) \in P  \}|$
            \State \textbf{return}($c \leq T$)
        \Else ~
            \textbf{return}($False$)
        \EndIf
    \EndIf
\EndFunction

\Statex

\Function{$\execute$}{$\tup{P,\prec}, op, i$}
    \If{$(op=\act{completed}(x, k))$}
        \State \textbf{return}($(\act{do}(x, k) \in P)$)
    \Else ~
        \textbf{return}($\bot$)
    \EndIf
\EndFunction
}
\end{algorithmic}
\end{multicols}
\caption{Functions $\valid()$ and $\execute()$ to implement a Do-All object given a threshold $T$ taken from a set $J$ of jobs to execute.}\label{alg:doall}
\end{algorithm}


\section{Conclusions}
In this paper we have formalized the notion of a validated object, decoupling the object operations and properties from the validation procedure. We have focused on two type of objects, satisfying different levels of consistency: the validated \order object, offering a total ordering of its operations, and its weaker variant, the validated \regular object. For both types, we have provided crash-tolerant implementations. Note that these implementations only attempt to prove that it is possible to implement different types of validated objects with and without consensus. Our objective was not to make them as efficient as possible; this is the subject of future work.

For validated \order objects, we further considered the persistent validity and persistent execution properties and their impact on the object's implementation. Our investigation has shown that $(i)$ in the absence of persistent validity, the object is as strong as consensus, and $(ii)$ persistent validity is not enough to implement a validated \order object without consensus; persistent execution was needed. An interesting future direction is to investigate whether there exists a weaker property than persistent execution, that together with persistent validity would yield consensus-free implementations of validated \order objects. 

Furthermore, this investigation could be extended for Byzantine failures. We believe that with certain adjustments, a Byzantine-tolerant implementation of validated \regular objects can be obtained from the one presented in Section~\ref{sec:implementing-regular-objects}. Observe that the consensus-based implementation of validated \order objects presented in Section~\ref{sec:consensus} can tolerate Byzantine failures, when a Byzantine-tolerant Atomic Broadcast service is used. Also, the negative result of Section~\ref{sec:negative} trivially applies to Byzantine failures. What remains to be investigated are the conditions under which it is possible to obtain Byzantine-tolerant consensus-free implementations of validated \order objects. Finally, other consistency levels for validated objects can be defined, beyond \regular and \order, and their implementability in different distributed system models be explored.

\bibliographystyle{plainurl}
\bibliography{sample,biblio}
\newpage
\begin{appendix}
\section{Proof of Lemma \ref{lemma:complete-validity}}
\label{sec:app:proofs}
{\bf Statement.} \sloppy{For each complete operation $op$ (issued by $i$), $\valid(\tup{P(op), \prec}, op, i)=True$. Moreover, $op$ returns in its response event the value $\execute(\tup{P(op), \prec}, op, i)$.}
\begin{proof}
The claim follows if we show that the set $P$ created in Line~\ref{p-def} of Code~\ref{alg:generic-apply} is the same as $P(op)$. 
Recall that $\timestamp(op)=(i, T_1, \ldots, T_k, \ldots, T_n)$.
Observe that in a ledger $L_k$ the timestamp $\timestamp=(k, T'_1, \ldots, T'_k, \ldots, T'_n)$ in the $j$th record in the ledger has
$T'_k=j-1$. Then, it holds that $P \subseteq P(op)$, because for each $k$, for each record $\langle \timestamp', op'\rangle \in G_k$, the timestamp $\timestamp'=(k, T'_1, \ldots, T'_k, \ldots, T'_n)$ satisfies that
$T'_k < T_k = |G_k|$.

Let us assume there is an operation $op' \in P(op)$ (hence, $op' \prec op$) and $op' \notin P$.
Assume 
$op'$ was issued by process $k$, and 
$\timestamp(op')=(k, T'_1, \ldots, T'_k, \ldots, T'_n)$.
Then, by linearizability of the ledgers and $op' \notin P$, $op'$ was appended in ledger $L_k$ after the $L_k.get()$ 
of $\apply(op,i)$ found $T_k$ records in the ledger.
Hence, $op'$ is the $j$th record in ledger $L_k$, where $j>T_k$. Note from Code~\ref{alg:generic-apply} that $T'_k=j-1$, since by well-formedness the $j$th operation of process $k$ finds $j-1$ records in ledger $L_k$.
Then, $T'_k \geq T_k$, and hence it cannot happen that $op' \prec op$.
\end{proof}

\section{Versioned Read/Write Objects}
\label{sec:Versioned}
A versioned object is a read/write object with the difference that each value written is associated 
with a version from a totally-ordered set of versions. A write operation succeeds only if it attempts to 
write a value with a version higher than any of the versions used by previous write operations; otherwise the write operation fails. 
In particular the object was introduced in \cite{Coverability16},
and supports two operations:
\begin{itemize}[leftmargin=5mm]
    \item $\act{write}(\tup{ver, v},x)$: process $i$ attempts to write value $v$ with version $ver$ on object $x$.
    \item $\act{read}(x)$: process $i$ attempts to
    read the latest value and version of the object $x$.
\end{itemize}

In the strict case only the writes that satisfy the total ordering may be executed and thus this will ensure a strict order on the version of the writes. Therefore, we will obtain a single consistent sequence of versions. On the other hand on the relaxed case, multiple writes promoting the same version may conflict allowing multiple writes to be executed. In such a case only some of those will succeed by the operation definition, thus ensuring the properties of  Coverability as presented in~\cite{Coverability16}.

\begin{algorithm}[h]
\caption{Functions $\valid()$ and $\execute()$ to implement a R/W versioned object.}
\label{alg:rwobject}
\begin{algorithmic}[1]
\Function{$\valid$}{$\tup{P,\prec}, op, i$}
   \If{$(i$ is not the issuer of $op)$}
       \State \textbf{return}($False$)
    \EndIf
    \If{$(op=\act{read}(x))$}
        \State \textbf{return}($True$)
    \Else \Comment{$op= \act{write}(\tup{ver,v}, x)$}
        \If{$op= \act{write}(\tup{ver,v}, x)$} 
            \State $ver_{max} \leftarrow \max{\{ver: \act{write}(\tup{ver,*},x) \in P  \}}$
            
            \State \textbf{return}($ver > ver_{max}$)
        \Else
            \State \textbf{return}($False$)
        \EndIf
    \EndIf
\EndFunction

\Statex

\Function{$\execute$}{$\tup{P,\prec}, op, i$}
    \If{$(op=\act{read}(x))$}
        \State $ver_{max}\leftarrow \max{\{ver: \act{write}(\tup{ver,*},x,j) \in P  \}}$
        \State $v_{max} \leftarrow \{v: \act{write}(\tup{ver_{max},v},x,j) \in P \}$
        \State \textbf{return}($\tup{ver_{max}, v_{max}}$)
    \Else \
        \State \textbf{return}($\bot$)
    \EndIf
\EndFunction
\end{algorithmic}
\end{algorithm}

A versioned object does not satisfy persistent validity, neither persistent execution. Thus, in order to implement the strict version of the object we use consensus. The following result holds.

\begin{theorem}
Code~\ref{alg:rwobject} combined with Code~\ref{alg:generic-apply} or with Code~\ref{alg:bab-apply}, implement 
the relax and strict versioned R/W object, respectively. 
\end{theorem}
\end{appendix}

\end{document}